\documentclass[journal,twocolumn]{IEEEtran}
\usepackage{color}
\usepackage{amsfonts}
\usepackage{amsthm}
\usepackage{verbatim}
\usepackage{multirow}
\usepackage{wrapfig}
\usepackage{subfig}

\usepackage{cite}      

\usepackage{graphicx}  

\usepackage{amsbsy}
\usepackage{amssymb}
\usepackage[normalem]{ulem}
\usepackage{amsmath,mathrsfs}   
\usepackage{cases}
\usepackage{braket}

\newcommand{\NT}{n}
\newcommand{\Nb}{N}
\newcommand{\cev}[1]{\reflectbox{\ensuremath{\vec{\reflectbox{\ensuremath{#1}}}}}}

\IEEEoverridecommandlockouts
\begin{document}

\title{Joint Bi-Directional Training of Nonlinear Precoders and Receivers in Cellular Networks}
\author{\authorblockN{Mingguang Xu, Dongning Guo, and Michael L. Honig\\}
  \authorblockA{Department of Electrical Engineering and Computer Science\\
Northwestern University\\
    2145 Sheridan Road, Evanston, IL 60208 USA
  }
    \thanks{This work was supported by
    DARPA under grant no. W911NF-07-1-0028, NSF under grant no. CCF-1018578, and a gift from Futurewei. This paper was
    presented in part at the 49th Annual Allerton
    Conference on Communication, Control, and Computing, Monticello, IL, September 2011.}}
\date{\today}

%


\maketitle

\begin{abstract}
Joint optimization of nonlinear precoders and receive filters is studied
for both the uplink and downlink in a cellular system.
For the uplink, the base transceiver station (BTS) receiver implements
successive interference cancellation, and for the downlink,
the BTS station pre-compensates for the interference with
Tomlinson-Harashima precoding (THP). Convergence of alternating
optimization of receivers and transmitters in a single cell
is established when filters are updated according
to a minimum mean squared error (MMSE)
criterion, subject to appropriate power constraints.
{\em Adaptive} algorithms are then introduced for updating
the precoders and receivers in the absence of channel state
information, assuming time-division duplex transmissions
with channel reciprocity. Instead of estimating the channels,
the filters are directly estimated according to a least
squares criterion via bi-directional training:
Uplink pilots are used to update the feedforward and feedback filters,
which are then used as interference pre-compensation
filters for downlink training of the mobile receivers.
Numerical results show that nonlinear filters can provide
substantial gains relative to linear filters with limited forward-backward iterations.
\end{abstract}


%

\section{Introduction}

As mobile broadband wireless networks evolve, increasingly ambitious targets
have been set for spectral efficiency.  Exploiting the capabilities of multiple antennas
at the mobiles and base transceiver station (BTS) is crucial to meet those goals.
An effective approach is to use spatial precoding at the transmitters
and filtering at receivers to adapt
to the variations of the multiple-input multiple-output (MIMO) channel
and interference conditions by optimizing a metric such as the sum rate,
signal-to-interference-and-noise ratio (SINR), or mean squared error (MSE).
Two challenges need to be addressed: First, the precoders and filters
should be {\em jointly} optimized; and second, the optimization should
be carried out without {\em a priori} knowledge of the channel and interference conditions.

While joint optimization of precoders and receivers in cellular networks
has attracted considerable attention in recent years
(e.g., \cite{ShiBer14TSP, KomTol13TSP, HunJoh09TSP, CodTol07TSP,
ShiSch07TSP, SchShi05VTC, PanWon04TWC, SerYen04TSP, WonMur03TWC}),
much of the work has assumed linear precoders for the downlink,
and has considered channel estimation
methods for estimating those precoders with unknown channels.
Here we consider joint optimization and adaptation of {\em non}linear precoders and
receivers. Specifically, we assume that
the users are successively decoded and cancelled at the BTS for the uplink, and that successive
interference pre-cancellation with Tomlinson-Harashima precoding is applied by the BTS for the downlink.
The mobiles use linear filtering and precoding.
A time-division duplex (TDD) network is assumed in order to exploit
the reciprocity between uplink and downlink channels.

The precoder and receiver filters are optimized according to a sum-MSE criterion
assuming MIMO fading channels.
This criterion is motivated by previous results showing that the minimum mean squared error (MMSE)
receiver with successive interference cancellation achieves the sum capacity
of the multiaccess channel \cite{TseVis05}, and that the
precoding structure considered (generalized decision-feedback equalizer with
interference pre-compensation) achieves the sum capacity
of the Gaussian broadcast channel \cite{YuCio04IT}.
Furthermore, the MSE criterion leads directly to adaptive approaches
for filter estimation in the presence of unknown channels and interference.

For a single cell we first show that alternating optimization
of precoders and receiver filters (for either the uplink or downlink)
converges to a local optimum.
Although this optimization must be applied separately
for the uplink and downlink to obtain the optimal precoders and receiver filters
for each direction, we observe that the filters optimized
for one direction achieve near-optimal performance for the other direction
as well when the uplink and downlink power constraints and
noise power are similar.\footnote{The MSE duality results
for broadcast and multiaccess channels in \cite{ShiSch07TSP,SchShi05VTC,HunJoh09TSP}
do not apply with different power constraints and noise powers
for the uplink and downlink,
hence cannot be used to compute the optimal precoders and receiver filters
for both uplink and downlink simultaneously.}
With this in mind we consider a simplified scheme in which
the uplink (downlink) receive filters are also used as
downlink (uplink) precoders. With this constraint
we again show that alternating optimization of the filters
at the BTS and mobiles converges to a fixed point.

In the absence of channel state information (CSI) the preceding results lead to an
iterative bi-directional training scheme in which uplink and downlink pilots
are alternately transmitted to estimate the BTS and mobile filters, respectively.
The estimated BTS forward/backward filters then nonlinearly precode both
the downlink data and pilots for interference pre-cancellation,
and the estimated mobile receiver filters precode both the uplink data and pilots.
With sufficient training, the estimated filters approach the corresponding
MMSE filters in which case training in each direction corresponds to an
MMSE update of the corresponding filters. These uplink/downlink updates
are then iterated until convergence, or a complexity constraint is reached.

This scheme extends the bi-directional training scheme
presented in \cite{ShiBer14TSP}, where the uplink/downlink pilots
are used to estimate linear filters/precoders without interference cancellation.
A key feature of these schemes is that the pilots are transmitted
{\em synchronously} in each direction (across all cells), and are used to estimate
the filters {\em directly} via a least squares criterion,
rather than directly estimating the channel. As discussed in \cite{ShiBer14TSP},
this scheme is {\em distributed} in the sense that each node (BTS or mobile)
autonomously computes its own filters from the received signal and pilots.
In contrast, prior schemes for precoding/filtering based on channel estimation
(e.g.,~\cite{LovHea08JSAC,GesHan10JSAC,WonMur03TWC,SerYen04TSP,PanWon04TWC,
SchShi05VTC,ShiSch07TSP,HunJoh09TSP})
are more suitable for centralized optimization
in which all CSI for direct- and cross-channels is passed
to a single remote location, which computes all filters and then
passes those back to the corresponding nodes.\footnote{See also the
two-way channel estimation methods presented in
\cite{Maz06Asilomar,SteSab08TWC,GomPap08ICC,WitTay08TSP,OsaMur09VTC,Zho10TWC}.
Extending those schemes to multi-cell scenarios requires the scheduling
of pilots across cells to avoid pilot pollution \cite{JosAshMar11}.}

Numerical results are presented that illustrate the performance of the
bi-directional training scheme as a function of the amount of training and
number of forward-backward iterations.
It is shown that the achievable rate with nonlinear filters approaches
the uplink sum capacity as the number of iterations increases.
For a fixed number of training symbols, the nonlinear filters significantly outperform linear filters when the number of iterations is limited
even though the nonlinear filters require more parameters to be estimated.
An example for a multi-cell scenario shows that more iterations are required
for the precoders and receive filters to converge than for a single cell.
Nonlinear filters again significantly outperform linear filters
with few iterations in the high-SNR regime.

Related work on joint precoder and receive filter design has been presented
in \cite{KomTol13TSP,CodTol07TSP,HanLeN13TSP,WonMur03TWC,PanWon04TWC,SerYen04TSP,ShiSch07TSP,SchShi05VTC,HunJoh09TSP,IltKim06TC}.
In \cite{WonMur03TWC,PanWon04TWC}, iterative algorithms are proposed for
simultaneous diagonalization of the downlink multiuser channels to enable
interference-free spatial-division multiplexing. In \cite{SerYen04TSP}, algorithms
are presented to jointly optimize linear filters and precoders for the uplink
using a sum-MSE criterion.
MSE duality results for broadcast and multiaccess channels are presented
in \cite{ShiSch07TSP,SchShi05VTC,HunJoh09TSP}. Duality ensures that any uplink
minimum MSE scheme 
has a downlink counterpart.
In \cite{ShiSch07TSP,SchShi05VTC},
the downlink sum-MSE minimization problem is 
solved by creating a ``virtual'' uplink channel, which has the
same transmit power and noise level as for the actual downlink channel.
However, those methods cannot be implemented in a distributed manner 
in the setting of this paper,
where the uplink and downlink channels 
may have different power constraints and noise levels. 
In \cite{IltKim06TC}, iterative MMSE beamforming algorithms are presented
to minimize the sum power subject to a set of SINR constraints,
where the precoder and receive filter updates are done sequentially
across users. In contrast,
all users update simultaneously in this paper.

\section{System Model}
Consider a single cell consisting of a BTS with
$\Nb$ antennas and $K$ users, where user $k$ has $N_k$ antennas.
The uplink channel between user $k$ and the BTS is denoted
by $\boldsymbol{H}_k$ of size $\Nb\!\times\!N_k$, each entry of which
is a unit-variance circularly symmetric complex Gaussian (CSCG)
random variable. Each user is restricted to one data stream,
although the algorithms to be presented are easily
extended to the scenario with multiple data streams per user.
The received signal vector at the BTS is
\begin{equation}
\vec{\boldsymbol{y}}=\sum_{k=1}^K\boldsymbol{H}_k\boldsymbol{v}_k\vec{x}_k+\vec{\boldsymbol{n}},\label{upsignal}
\end{equation}
where $\vec{x}_k$ is the data symbol from user $k$
and $\mathbb{E}\big[|\vec{x}_k|^2\big]\!=\!1$,
$\boldsymbol{v}_k$ is the $N_k \times 1$ precoder
with $\|\boldsymbol{v}_k\|^2\!\leq\!P_k$, where $P_k$ is
the power constraint for user $k$,
and $\vec{\boldsymbol{n}}$ is the CSCG noise vector
each entry of which has zero mean and variance $\sigma^2$.

Similarly, assuming TDD with channel reciprocity,
the downlink received signal at user $k$ can be written as
\begin{equation}
\cev{\boldsymbol{y}_k}=\boldsymbol{H}_k^\dagger\bigg(
\sum_{i=1}^K\boldsymbol{t}_i\cev{x_i}\bigg)+
\cev{\boldsymbol{n}}_k,\label{eq:downrecSig}
\end{equation}
where $\cev{x_i}$ is the data symbol for user $i$,
$\boldsymbol{t}_i$ is the corresponding transmit beam,
and $\cev{\boldsymbol{n}}_k$ denotes the CSCG noise vector with $k$th entry having
variance $\sigma^2_k$. The filters satisfy the BTS power constraint
$\sum_{i=1}^K\!\|\boldsymbol{t}_i\|^2\!\leq\!P$. 

\section{Joint Optimization of Precoders and Receivers}
\label{sec:linearup}
We first discuss joint optimization of transmit precoders
and receive filters at a central controller with perfect CSI.
Adaptive versions with training are subsequently
discussed in Section~\ref{s:training}.

\subsection{Uplink Filters}
\label{s:linear}
First consider the case where the received signal in (\ref{upsignal})
is processed by linear filters
$\boldsymbol{G}\!=\!\big[\boldsymbol{g}_1,\ldots,\boldsymbol{g}_K\big]$,
where $\boldsymbol{g}_k$ denotes the receive filter for user $k$.
The estimated signal for user $k$ is given by
$\hat{x}_k\!=\!\boldsymbol{g}_k^\dagger\vec{\boldsymbol{y}}$
and the corresponding MSE for user $k$ is
\begin{align}
\vec{\varepsilon}(k)&\!=\!\mathbb{E}\big[|\vec{x}_k-\hat{x}_k|^2\big]\\
&\!=\!1\!+\!\sigma^2\|\boldsymbol{g}_k\|^2\!-\!2\Re\big\{\boldsymbol{g}_k^\dagger\boldsymbol{H}_k\boldsymbol{v}_k\big\}\!+\!\sum_{i=1}^{K}\big|\boldsymbol{g}_k^\dagger\boldsymbol{H}_i\boldsymbol{v}_i\big|^2,\label{upMSE}
\end{align}
where $\Re\{\cdot\}$ takes the real part of a complex number.
The optimization problem is stated as:
\begin{subequations}
\label{upproblem}
\begin{align}
&\underset{\{\boldsymbol{g}_k,\boldsymbol{v}_k\}}{\operatorname{\text{minimize}}}~~\sum_{k=1}^K\vec{\varepsilon}(k)\\
&\text{subject to}~~\|\boldsymbol{v}_k\|^2\leq P_k,~k=1,\ldots,K.\label{powercon}
\end{align}
\end{subequations}


\newtheorem{prop}{Proposition}
\begin{prop}
\label{prop:fullpow}
The sum MMSE is minimized when all users transmit with full power.
\end{prop}

The proof is given in the appendix.
Hence the decrease in sum MSE for a particular user $k$ due to an
increase in power $P_k$ dominates the increase in sum MSE
due to the increase in interference at neighboring
receivers. We will assume that each user transmits at the maximum power
so that the constraints (\ref{powercon}) are binding.


\begin{figure*}
\center
\includegraphics[width=0.75\textwidth]{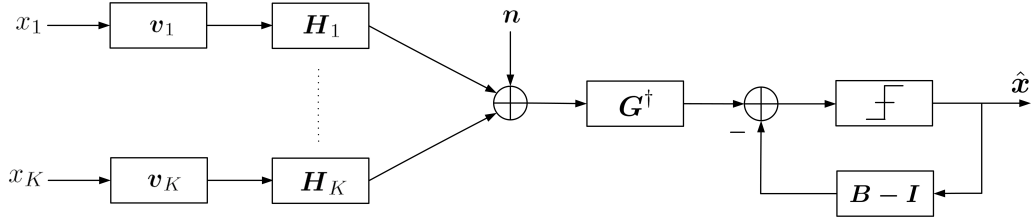}
\caption{\small Decision-feedback detector in the uplink.}
 \label{uplink}
\end{figure*}

The Lagrangian for problem~(\ref{upproblem}) is
\begin{equation}
L(\boldsymbol{g}_1,\ldots,\boldsymbol{g}_K, \boldsymbol{v}_1,\ldots,\boldsymbol{v}_K)
=\sum_{k=1}^K\vec{\varepsilon}(k)+\sum_{k=1}^K\mu_k\big(\|\boldsymbol{v}_k\|^2-P_k\big),\label{Lagrangian}
\end{equation}
where $\mu_k$ is the Lagrange multiplier associated with the power constraint for user $k$.
Hence the optimal solution to problem (\ref{upproblem}) must satisfy:
\begin{align}
\boldsymbol{g}_k & =\bigg(\displaystyle{\sum_{i=1}^{K}}
\boldsymbol{H}_i\boldsymbol{v}_i\boldsymbol{v}_i^\dagger\boldsymbol{H}_i^\dagger
+\sigma^2\boldsymbol{I}\bigg)^{-1}\boldsymbol{H}_k\boldsymbol{v}_k \label{eq:upLnFilter1}\\
\boldsymbol{v}_k & =\bigg(\displaystyle{\sum_{i=1}^{K}}
\boldsymbol{H}_k^\dagger\boldsymbol{g}_i\boldsymbol{g}_i^\dagger\boldsymbol{H}_k
+\mu_k\boldsymbol{I}\bigg)^{-1}\boldsymbol{H}_k^\dagger\boldsymbol{g}_k
\label{eq:upLnFilter2}
\end{align}
for $k=1,\dots,K$, where $\mu_k$ in (\ref{eq:upLnFilter2})
is chosen to satisfy (\ref{powercon}).

The nonlinear successive cancellation filter at the BTS is the
decision-feedback detector (DFD) shown in Fig.~\ref{uplink}.
The received signal is input to the linear feedforward filter $\boldsymbol{G}$,
and the feedback filter $\boldsymbol{B}-\boldsymbol{I}$ implements
the successive interference cancellation.  Assuming the decoding order
is from user $1$ to user $K$, the feedback matrix $\boldsymbol{B}$
is lower triangular with diagonal elements equal to one.
The estimated symbol for user $k$ is then
\begin{equation}
\begin{split}
\hat{x}_k=\boldsymbol{g}_k^\dagger\boldsymbol{H}_k\boldsymbol{v}_k
&\vec{x}_k+\boldsymbol{g}_k^\dagger\bigg(\sum_{i=k+1}^{K}\boldsymbol{H}_i\boldsymbol{v}_i\vec{x}_i+\vec{\boldsymbol{n}}\bigg)\\
&+\sum_{i=1}^{k-1}\bigg(\boldsymbol{g}_k^\dagger\boldsymbol{H}_i\boldsymbol{v}_i\vec{x}_i-b_{ki}Q(\hat{x}_i)\bigg),\label{dfesignal}
\end{split}
\end{equation}
where $b_{ki}$ denotes the ($k$,$i$)-th element of $\boldsymbol{B}$
and $Q(\cdot)$ denotes the slicer function in the forward loop.
For user $k$, if decoding of the previous symbols is correct,
i.e., $Q(\hat{x}_i)\!=\!\vec{x}_i,~i\!=\!1,\ldots,k\!-\!1$,
and $b_{ki}=\boldsymbol{g}_k^\dagger\boldsymbol{H}_i\boldsymbol{v}_i$,
then the interference from users 1 to user $k\!-\!1$ can be removed.
The achievable MSE for user $k$ is then
\begin{align}
\vec{\varepsilon}_{\text{dfd}}(k)\!=\!1\!+\!\sigma^2\|\boldsymbol{g}_k\|^2\!\!-\!\!2\Re\big\{\boldsymbol{g}_k^\dagger\boldsymbol{H}_k\boldsymbol{v}_k\big\}\!\!+\!\!\sum_{i=k}^{K}\big|\boldsymbol{g}_k^\dagger\boldsymbol{H}_i\boldsymbol{v}_i\big|^2\!.
\end{align}
The problem becomes
\begin{subequations}
\label{dfeproblem}
\begin{align}
&\underset{\{\boldsymbol{g}_k,\boldsymbol{v}_k\}}{\operatorname{\text{minimize}}}~~
\sum_{k=1}^K\vec{\varepsilon}_{\text{dfd}}(k)\\
&\text{subject to}~~\|\boldsymbol{v}_k\|^2=P_k,~k=1,\ldots,K.\label{con:fullpow}
\end{align}
\end{subequations}
In this scenario it may not be optimal (in the sum-MSE sense)
for all users to transmit with full power.
We will nevertheless assume that the power constraints are binding
since this simplifies the problem, and
because transmitting at full power is known to achieve the uplink sum capacity.


Writing the Lagrangian as before, the corresponding first-order conditions are
\begin{align}
\boldsymbol{g}_k & =\bigg(
 \displaystyle{\sum_{i=k}^{K}}
\boldsymbol{H}_i\boldsymbol{v}_i\boldsymbol{v}_i^\dagger\boldsymbol{H}_i^\dagger
+\sigma^2\boldsymbol{I}\bigg)^{-1}\boldsymbol{H}_k\boldsymbol{v}_k\label{eq:upNlnFilter1}\\
\boldsymbol{v}_k & =\bigg(\displaystyle{\sum_{i=1}^{k}}
\boldsymbol{H}_k^\dagger\boldsymbol{g}_i\boldsymbol{g}_i^\dagger\boldsymbol{H}_k
+\mu_k\boldsymbol{I}\bigg)^{-1}\boldsymbol{H}_k^\dagger\boldsymbol{g}_k\label{eq:upNlnFilter2}
\end{align}
for $k=1,\dots,K$, where the Lagrange multipliers $\mu_k$'s are chosen
to enforce the power constraints (\ref{con:fullpow}).
The summation indices differ from (\ref{eq:upLnFilter1})
and (\ref{eq:upLnFilter2}) due to successive cancellation.

\subsection{Downlink Filters} 
\label{s:down}
The expressions for the downlink filters can be obtained in a
similar fashion as for the uplink.
The uplink and downlink problems are different because the SNRs
for the uplink and downlink channels are different in general,
and because the power constraints are on individual users in the uplink,
while the constraint is on the sum over all users in the downlink.

For the downlink, the estimated symbol for user $k$ is
$\hat{x}_k\!=\!\boldsymbol{r}_k^\dagger\cev{\boldsymbol{y}_k}$
where $\boldsymbol{r}_k$ is the receiver filter.
The corresponding MSE is
\begin{align}
\cev{\varepsilon}(k)\!=\!1\!+\!\sigma_k^2\|\boldsymbol{r}_k\|^2\!-\!2\Re\big\{\boldsymbol{r}_k^\dagger\boldsymbol{H}_k^\dagger\boldsymbol{t}_k\big\}\!+\!\sum_{i=1}^{K}\big|\boldsymbol{r}_k^\dagger\boldsymbol{H}_k^\dagger\boldsymbol{t}_i\big|^2.
\end{align}
The downlink sum-MSE minimization problem is
\begin{subequations}
\label{dproblem}
\begin{align}
&\underset{\{\boldsymbol{r}_k,\boldsymbol{t}_k\}}{\operatorname{\text{minimize}}}~~\sum_{k=1}^K\cev{\varepsilon}(k)\\
&\text{subject to}~~\sum_{k=1}^K\|\boldsymbol{t}_k\|^2=P.\label{con:downfullpow}
\end{align}
\end{subequations}
The corresponding optimal filters satisfy:
\begin{align}
\boldsymbol{t}_k & =\bigg(\displaystyle{\sum_{i=1}^{K}}
\boldsymbol{H}_i\boldsymbol{r}_i\boldsymbol{r}_i^\dagger\boldsymbol{H}_i^\dagger
+\mu\boldsymbol{I}\bigg)^{-1}\boldsymbol{H}_k\boldsymbol{r}_k \label{eq:downLnFilter2}\\
\boldsymbol{r}_k & =\bigg(\displaystyle{\sum_{i=1}^{K}}
\boldsymbol{H}_k^\dagger\boldsymbol{t}_i\boldsymbol{t}_i^\dagger\boldsymbol{H}_k
+\sigma_k^2\boldsymbol{I}\bigg)^{-1}
\boldsymbol{H}_k^\dagger\boldsymbol{t}_k \label{eq:downLnFilter1}
\end{align}
for $k=1,\dots,K$, where the Lagrange multiplier $\mu$
is chosen to satisfy \eqref{con:downfullpow}.

With nonlinear filtering at the BTS \cite{Costa83IT,WeiSte06IT},
i.e., dirty paper coding,
assuming the encoding order is from user $K$ to user $1$,
the signal received by user $k$ can be written as
\begin{equation}
\cev{\boldsymbol{y}_k}=\boldsymbol{H}_k^\dagger\bigg(\sum_{i=1}^k\boldsymbol{t}_i\cev{x_i}\bigg)+\cev{\boldsymbol{n}}_k.
\label{eq:downrecSigNln}
\end{equation}
The corresponding MSE for user $k$ at the output of the
linear receiver $\boldsymbol{r}_k$ is then
\begin{align}
\cev{\varepsilon}_{\text{ipc}}(k)\!=\!1\!+\!\sigma_k^2\|\boldsymbol{r}_k\|^2\!\!-\!\!2\Re\big\{\boldsymbol{r}_k^\dagger\boldsymbol{H}_k^\dagger\boldsymbol{t}_k\big\}\!+\!\sum_{i=1}^{k}\big|\boldsymbol{r}_k^\dagger\boldsymbol{H}_k^\dagger\boldsymbol{t}_i\big|^2.
\end{align}
The downlink sum-MSE minimization problem is
\begin{subequations}
\label{dproblem_ipc}
\begin{align}
&\underset{\{\boldsymbol{r}_k,\boldsymbol{t}_k\}}{\operatorname{\text{minimize}}}~~\sum_{k=1}^K\cev{\varepsilon}_{\text{ipc}}(k)\\
&\text{subject to}~~\sum_{k=1}^K\|\boldsymbol{t}_k\|^2=P.\label{con:downfullpow_ipc}
\end{align}
\end{subequations}
The corresponding optimal filters satisfy:
\begin{align}
\boldsymbol{t}_k & =\bigg(\displaystyle{\sum_{i=k}^{K}}\boldsymbol{H}_i\boldsymbol{r}_i\boldsymbol{r}_i^\dagger\boldsymbol{H}_i^\dagger
+\mu\boldsymbol{I}\bigg)^{-1}\boldsymbol{H}_k\boldsymbol{r}_k \label{eq:downnLnFilter2}\\
\boldsymbol{r}_k & =\bigg(\displaystyle{\sum_{i=1}^{k}}\boldsymbol{H}_k^\dagger\boldsymbol{t}_i\boldsymbol{t}_i^\dagger\boldsymbol{H}_k
+\sigma_k^2\boldsymbol{I}\bigg)^{-1}\boldsymbol{H}_k^\dagger\boldsymbol{t}_k \label{eq:downnLnFilter1}
\end{align}
for $k=1,\dots,K$, where the Lagrange multiplier $\mu$
is chosen to satisfy \eqref{con:downfullpow}.

\section{Iterative Computation}
\subsection{Separate Uplink/Downlink Optimization}
\label{Sec:iterativeSeparate}
Focusing on the uplink the precoders and filters
in (\ref{eq:upLnFilter1})--(\ref{eq:upLnFilter2}) are coupled,
evading an explicit solution in terms of the channels only.
Moreover, problem~\eqref{upproblem} is non-convex when $N_k\!>\!1$,
thus there is no known efficient numerical algorithm for finding the global optimum.\footnote{Indeed multiple local optima have been observed in numerical experiments.  However, if $N_k\!=\!1$, or if $N_k\!>\!1$ and there is no constraint on the number of data streams per user, then this problem can be transformed into a convex program~\cite{SerYen04TSP,ShiSch07TSP}.}
As pointed out in~\cite{SerYen04TSP}, the precoders and filters can be computed iteratively.
With fixed transmit precoders $\{\boldsymbol{v}_k\}$, the Lagrangian (\ref{Lagrangian}) is strictly convex in receive filters $\{\boldsymbol{g}_k\}$, and (\ref{eq:upLnFilter1}) gives the optimal solution.  Conversely, with fixed $\{\boldsymbol{g}_k\}$, we compute $\{\boldsymbol{v}_k\}$ and $\{\mu_k\}$ together from~\eqref{eq:upLnFilter2} with the requirement that $\|\boldsymbol{v}_k\|^2\!=\!P_k$ for every $k$.  The search for each $\mu_k$ can be efficient because the optimization of $(\boldsymbol{v}_k, \mu_k)$ is decoupled across users.


Because $\mu_k$ is chosen to guarantee $\|\boldsymbol{v}_k\|^2\!\!=\!\!P_k$, the Lagrangian in (\ref{Lagrangian}) is always equal to the sum-MSE.  Thus the Lagrangian as well as the sum-MSE monotonically decrease with each update, and must converge since they are lower bounded by zero.
In case multiple $\mu_k$'s enforce $\|\boldsymbol{v}_k\|^2\!=\!P_k$,
we choose the one that yields the minimum Lagrangian (\ref{Lagrangian}).
If multiple $\mu_k$'s achieve the same minimum, then we can randomly
choose one.
Since the preceding procedure dictates a one-to-one mapping between $\{\boldsymbol{v}_k\}$ and $\{\boldsymbol{g}_k\}$, it is not difficult to show that the set of precoders and filters $\{\boldsymbol{v}_k,\boldsymbol{g}_k\}$ must also converge to a fixed point.
The fixed point must be a local optimum for problem (\ref{upproblem}),
although it may not be globally optimal.

Similar to the linear case, the iterative approach previously described
can also be applied to compute (\ref{eq:upNlnFilter1})--(\ref{eq:upNlnFilter2})
and the convergence of the filters follows from the monotonic decrease of the Lagrangian and the one-to-one mapping between the precoders and the receiver filters.
For the downlink the set of linear filters can be computed by
iterating (\ref{eq:downLnFilter1}) and (\ref{eq:downLnFilter2}),
and the set of nonlinear filters can be computed by iterating
(\ref{eq:downnLnFilter2}) and (\ref{eq:downnLnFilter1}).
According to similar arguments as for the uplink,
all the filters will again converge to a fixed point.


\subsection{Simultaneous Uplink/Downlink Optimization}
\label{Sec:uniApproach}

The iterative optimization procedure in Section~\ref{Sec:iterativeSeparate}
was applied separately to compute the set of uplink filters, given
by a solution to (\ref{upproblem}), and the set of downlink filters,
given by a solution to (\ref{dproblem}).
Here we present a suboptimal iterative algorithm that
computes all uplink/downlink precoders and filters simultaneously.
Numerical results show that the performance of this algorithm
is very close to that of separate uplink and downlink optimizations
with practically relevant SNRs.

Let $\boldsymbol{v}$ and $\boldsymbol{g}$ denote the uplink precoders and receive filters as in Section~\ref{sec:linearup}, and let $\boldsymbol{t}$ and $\boldsymbol{r}$ denote the downlink precoders and receive filters.
Taking linear filters as an example, we apply the following four steps
in each iteration:
\begin{enumerate}
\item Let each uplink $\boldsymbol{v}_i$ be the downlink receiver $\boldsymbol{r}_i$, and set $\mathbb{E}\big[|\vec{x}_k|^2\big]\!=\!\rho$.
\item Compute $\boldsymbol{g}_k=\bigg({\displaystyle\sum_{i=1}^{K}}\boldsymbol{H}_i \boldsymbol{v}_i  \boldsymbol{v}_i^\dagger\boldsymbol{H}_i^\dagger+\frac{\sigma^2}{\rho}\boldsymbol{I}\bigg)^{-1}\boldsymbol{H}_k\boldsymbol{v}_k$.
\item Let each downlink  $\boldsymbol{t}_i$ be the uplink receiver $\boldsymbol{g}_i$, and set $\mathbb{E}\big[|\cev{x_k}|^2\big]\!=\!\beta$.
\item Compute $\boldsymbol{r}_k=\bigg({\displaystyle\sum_{i=1}^{K}}\boldsymbol{H}_k^\dagger \boldsymbol{t}_i\boldsymbol{t}_i^\dagger\boldsymbol{H}_k
 +\frac{\sigma_k^2}{\beta}\boldsymbol{I}\bigg)^{-1}\boldsymbol{H}_k^\dagger\boldsymbol{t}_k$.
\end{enumerate}
That is, we use the same update formulas as for the uplink and downlink
optimizations discussed in Section~\ref{sec:linearup},
but constrain the precoders and receive filters at each terminal
to be the same. The uplink and downlink powers $\rho$ and $\beta$
are assumed to be given {\em a priori}.  This assumption is needed to prove the following convergence result, but will be relaxed in the next section.

\begin{prop}
  For any fixed positive $\rho$ and $\beta$, the precoders and filters $\big\{\boldsymbol{v}_k,\boldsymbol{g}_k,\boldsymbol{t}_k,\boldsymbol{r}_k\big\}$
computed according to the preceding iterative algorithm
converge to a fixed point as the number of iterations goes to infinity.
\end{prop}
\begin{proof}
We provide a positive potential function of the precoders and filters
and show that this function decreases every iteration.
The uplink MSE, normalized by $\rho$, for user $k$ is
\begin{equation}
\vec{\varepsilon}(k,\rho)\!=\!1\!+\frac{\!\sigma^2\|\boldsymbol{g}_k\|^2}{\rho}\!-\!2\Re\big\{\boldsymbol{g}_k^\dagger\boldsymbol{H}_k\boldsymbol{v}_k\big\}\!+\!\sum_{i=1}^{K}\big|\boldsymbol{g}_k^\dagger\boldsymbol{H}_i\boldsymbol{v}_i\big|^2
\end{equation}
and the downlink MSE, normalized by $\beta$, for user $k$ is
\begin{equation}
\cev{\varepsilon}(k,\beta)\!=\!1\!+\frac{\!\sigma^2_k\|\boldsymbol{r}_k\|^2}{\beta}\!-\!2\Re\big\{\boldsymbol{r}_k^\dagger\boldsymbol{H}_k^\dagger\boldsymbol{t}_k\big\}\!+\!\sum_{i=1}^{K}\big|\boldsymbol{r}_k^\dagger\boldsymbol{H}^\dagger_k\boldsymbol{t}_i\big|^2.
\end{equation}

We define the functions:
\begin{equation}
\vec{\phi}(\boldsymbol{g}_1,\ldots,\boldsymbol{g}_K,\boldsymbol{v}_1,\ldots,\boldsymbol{v}_K)\!\triangleq\!\sum_{k=1}^K\bigg(\vec{\varepsilon}(k,\rho)\!+\!\frac{\sigma^2_k\|\boldsymbol{v}_k\|^2}{\beta}\bigg)
\end{equation}
and
\begin{equation}
\cev{\phi}(\boldsymbol{r}_1,\ldots,\boldsymbol{r}_K,\boldsymbol{t}_1,\ldots,\boldsymbol{t}_K)\!\triangleq\!\sum_{k=1}^K\bigg(\cev{\varepsilon}(k,\beta)\!+\!\frac{\sigma^2\|\boldsymbol{t}_k\|^2}{\rho}\bigg).
\end{equation}
Since $\boldsymbol{v}_k\!=\!\boldsymbol{r}_k$ and $\boldsymbol{t}_k\!=\!\boldsymbol{g}_k$, it can be verified that $\vec{\phi}=\cev{\phi}$ for any given set of filters $\{\boldsymbol{v}_k,\boldsymbol{g}_k,\boldsymbol{t}_k,\boldsymbol{r}_k\},k\!=\!1,\ldots,K$. The convergence of $\vec{\phi}$ (or $\cev{\phi}$) can then be verified by observing that $\vec{\phi}$ (or $\cev{\phi}$) monotonically decreases with an update for any $\boldsymbol{g}_k$ at the BTS or any $\boldsymbol{r}_k$ for a particular user, and it is lower bounded by zero. Since in each step the mapping from $\{\boldsymbol{v}_k\}$ to $\{\boldsymbol{g}_k\}$ or from $\{\boldsymbol{t}_k\}$ to $\{\boldsymbol{r}_k\}$ is unique, it is easy to show that the filters must converge to a fixed point.
\end{proof}

With similar arguments this proposition can be extended to nonlinear filters
and to interference networks with multiple users and multiple cells.
In the case of interference networks, the potential function can be chosen
as the the sum of the achievable MSE's plus the scaled norms of
the precoders at all transmitters.



Ideally, the parameters $\rho$ and $\beta$ are such that the transmit power constraints at the BTS and mobiles are satisfied at convergence.  The values are of course unknown {\em a priori}.
A straightforward power scaling method is to normalize the precoders directly after each iteration so that the transmit power constraint is met.
This is analogous to the max-SINR algorithm proposed in \cite{GomCad11IT} for interference networks. Although we cannot prove the convergence of the corresponding optimization algorithm, convergence has always been observed for numerical examples.  In fact the performance of this algorithm is very close to that of separate optimization of the uplink problem~\eqref{upproblem} and its downlink counterpart.

\section{Distributed Optimization via Bi-Directional Training} 
\label{s:training}

In this section, we assume the channels are unknown to the BTS and mobiles {\em a priori}.
We propose a bi-directional training scheme for adapting
the precoders at the mobiles and the filters at the BTS with the goal of minimizing the uplink sum-MSE.
This bi-directional training scheme can be implemented in a distributed manner.
Specifically, estimates of (\ref{eq:upLnFilter1}) and (\ref{eq:upNlnFilter1})
can be obtained by training in the uplink, and estimates of
(\ref{eq:upLnFilter2}) and (\ref{eq:upNlnFilter2}) can be obtained
by training in the downlink.
A similar scheme can be applied to the downlink;
since the uplink and downlink problems are similar,
we focus on the uplink in this section.
\begin{figure*}
  \center
  \includegraphics[width=0.75\textwidth]{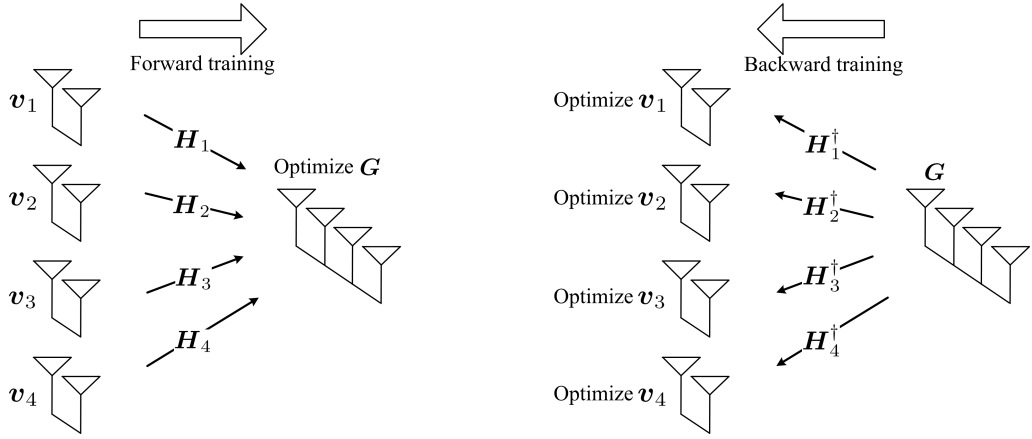}
  \caption{\small Forward--backward (bi-directional) training.}
   \label{Bidirectional}
\end{figure*}

\subsection{The Uplink with Linear Filtering}
\label{sec:LinearTrain}
We first observe the equivalence between iteratively computing
(\ref{eq:upLnFilter1})-(\ref{eq:upLnFilter2}) and bi-directional {\em optimization},
in which the following steps are iterated, as illustrated in Fig. \ref{Bidirectional}:
\begin{enumerate}
\item {\em Forward optimization:} Fix the uplink precoders
$\boldsymbol{v}_k$ and select the receive filters $\{\boldsymbol{g}_k\}$
to minimize the output MSE.
\item {\em Backward optimization:} Reverse the direction of transmission,
fix the downlink precoders at the BTS as $\{\boldsymbol{g}_k\}$ and select
the receive filter $\boldsymbol{v}_k$ to minimize the output MSE
subject to a norm constraint.
\end{enumerate}

In the first step, the optimal linear receive filters are given by (\ref{eq:upLnFilter1}).
In the second step,
let $\mathbb{E}\big[|\cev{x_k}|^2\big]\!=\!\gamma$,
where $\gamma$ is chosen such that the total downlink transmit power
$\sum_{k=1}^{K}\!\gamma\|\boldsymbol{g}_k\|^2$ does not exceed the constraint.
By channel reciprocity, the signal received by user $k$ is given by
(\ref{eq:downrecSig}) with $\boldsymbol{t}_i$ replaced by $\boldsymbol{g}_i$.
The MSE, normalized by $\gamma$, with receive filter $\boldsymbol{v}_k$ is then
\begin{align}
\cev{\varepsilon}(k,\gamma)\!\!=\!\!1\!+\!\frac{\sigma_k^2\|\boldsymbol{v}_k\|^2}{\gamma}\!\!-\!\!2\Re\big\{\boldsymbol{v}_k^\dagger\boldsymbol{H}_k^\dagger\boldsymbol{g}_k\big\}\!\!+\!\!\sum_{i=1}^{K}\!\big|\boldsymbol{v}_k^\dagger\boldsymbol{H}_k^\dagger\boldsymbol{g}_i\big|^2.\label{dlMSE}
\end{align}

In the second step we wish to
\begin{subequations}
\label{dlMSEproblem}
\begin{align}
&\underset{\boldsymbol{v}_k}{\operatorname{\text{minimize}}}~~
\cev{\varepsilon}(k,\gamma)\label{downMSE}\\
&\text{subject to}~~\|\boldsymbol{v}_k\|^2=P_k.\label{con:fullpowtrainDown}
\end{align}
\end{subequations}
By constraining the norm of $\boldsymbol{v}_k$, we account for the uplink power constraint.
Namely, the solution to problem (\ref{dlMSEproblem}) is given by
\begin{equation}
\boldsymbol{v}_k=\bigg(\displaystyle{\sum_{i=1}^{K}}\boldsymbol{H}_k^\dagger\boldsymbol{g}_i\boldsymbol{g}_i^\dagger\boldsymbol{H}_k
+\bigg(\frac{\sigma_k^2}{\gamma}\!+\!\nu\bigg)\boldsymbol{I}\bigg)^{-1}\boldsymbol{H}_k^\dagger\boldsymbol{g}_k,
\end{equation}
where the Lagrange multiplier $\nu$ is chosen to satisfy (\ref{con:fullpowtrainDown}).
If there are multiple $\nu$'s that satisfy (\ref{con:fullpowtrainDown}),
then the one which minimizes (\ref{downMSE}) should be chosen.
Since this solution is the same as in (\ref{eq:upLnFilter2}),
the previous iterative procedure for minimizing the Lagrangian in (\ref{Lagrangian})
is equivalent to bi-directional optimization.
Note that with the constraint (\ref{con:fullpowtrainDown}),
the term $\sigma_k^2\|\boldsymbol{v}_k\|^2/\gamma$ in (\ref{dlMSE})
is a constant so that the solution to problem (\ref{dlMSEproblem})
does not depend on the noise level at user $k$.
This is expected since problem (\ref{upproblem}) is to minimize
the uplink sum-MSE and the solution should not depend on the noise level
at the user side.

%

With unknown channels the MSE optimization criterion can be replaced
by a least squares cost function given the transmitted pilot symbols.
In the uplink, we assume that the users {\em synchronously} transmit
sequences of $\NT$ training symbols given by the matrix $\vec{\boldsymbol{S}}^T$
where $\vec{\boldsymbol{S}}\!=\!
[\vec{\boldsymbol{s}}_1^T,\ldots,\vec{\boldsymbol{s}}_K^T]$
and $\vec{\boldsymbol{s}}_k$ is the $1\times\NT$ row vector
containing the training symbols for user $k$.
The received signal at the BTS is then given by (\ref{upsignal}) with $\vec{x}_k\!=\!\vec{s}_k$. The corresponding sequence of estimated symbols is $\hat{\boldsymbol{s}}_k\!=\!\boldsymbol{g}_k^\dagger\vec{\boldsymbol{Y}}$ where $\vec{\boldsymbol{Y}}\!=\!\big[\vec{\boldsymbol{y}}(1),\ldots,\vec{\boldsymbol{y}}(\NT)\big]$. The receive filter is then chosen to minimize the cost function $\|\vec{\boldsymbol{s}}_k\!-\!\boldsymbol{g}_k^\dagger\vec{\boldsymbol{Y}}\|^2$, and the solution is given by
\begin{equation}
\boldsymbol{g}_k=\big(\vec{\boldsymbol{Y}}\vec{\boldsymbol{Y}}^\dagger\big)^{-1}\vec{\boldsymbol{Y}}\vec{\boldsymbol{s}}_k^\dagger.
\end{equation}

To estimate the uplink precoders, a similar training scheme
can be applied on the downlink where the training symbols
are passed through the precoders $\{\boldsymbol{g}_k\}$ and superposed
before transmission at the BTS. The received signal at user $k$
is then given by (\ref{eq:downrecSig}) with $\cev{x_k}\!=\!\cev{s_k}$ and $\boldsymbol{t}_i$ replaced by $\boldsymbol{g}_i$. At each user,
the cost function is given by $\|\cev{\boldsymbol{s}_k}\!-\!\boldsymbol{v}_k^\dagger\cev{\boldsymbol{Y}_k}\|^2/\gamma\!+\!\mu_k\big(\|\boldsymbol{v}_k\|^2\!-\!P_k\big)$ where an additional norm constraint is added, and $\cev{\boldsymbol{Y}_k}\!=\!\big[\cev{\boldsymbol{y}_k}(1),\ldots,\cev{\boldsymbol{y}_k}(\NT)\big]$. The solution is given by
\begin{equation}
\boldsymbol{v}_k=\big(\cev{\boldsymbol{Y}_k}\cev{\boldsymbol{Y}_k}^\dagger+\mu_k\boldsymbol{I}\big)^{-1}\cev{\boldsymbol{Y}_k}\cev{\boldsymbol{s}_k}^\dagger,\label{LSuser}
\end{equation}
where $\mu_k$ is chosen to satisfy the norm constraint $\|\boldsymbol{v}_k\|^2\!=\!P_k$. If multiple $\mu_k$'s satisfy the constraint, then the one which minimizes the cost function should be chosen.

\subsection{The Uplink with Nonlinear Filtering}
\label{sec:nonlinear}
We next present a bi-directional training scheme for implementing
the updates (\ref{eq:upNlnFilter1})--(\ref{eq:upNlnFilter2})
with nonlinear filters in the absence of channel information.
We assume a DFD in the uplink, and the analogous interference
pre-compensation scheme for the downlink
as illustrated in Fig. \ref{backward_training}.
Downlink interference is pre-compensated in the {\em reverse} order
as for the uplink, so that the effective downlink channel
is triangular. (That is, for the downlink user $k$ receives no interference
from users $k+1, \dots,K$, whereas for the uplink it receives
no interference from users $1,\dots, k-1$ \cite{YuCio04IT}.)

As before, we first relate alternating optimization of
(\ref{eq:upNlnFilter1})-(\ref{eq:upNlnFilter2})
to bi-directional optimization.
Namely, fixing the uplink precoders $\{\boldsymbol{v}_k\}$,
the optimal feedforward filters $\{\boldsymbol{g}_k\}$
are given by (\ref{eq:upNlnFilter1}).
To optimize the uplink precoders, consider downlink transmission
with perfect interference pre-compensation
using precoders $\{\boldsymbol{g}_k\}$ at the BTS,
and with encoding order from user $K$ to user $1$.
The received symbol at user $k$ can be written as
(\ref{eq:downrecSigNln})
with $\boldsymbol{t}_i\!=\!\boldsymbol{g}_i$.
With receive filter $\boldsymbol{v}_k$,
the MSE (normalized by $\gamma$) at user $k$ is given by
\begin{align}
\cev{\varepsilon}_{\text{ipc}}\!(k,\!\gamma)\!\!=\!1\!\!+\!\!\frac{\sigma_k^2\|\boldsymbol{v}_k\|^2}{\gamma}\!\!-\!\!2\Re\big\{\boldsymbol{v}_k^\dagger\!\boldsymbol{H}_k^\dagger\boldsymbol{g}_k\big\}\!\!+\!\!\sum_{i=1}^{k}\!\big|\boldsymbol{v}_k^\dagger\!\boldsymbol{H}_k^\dagger\boldsymbol{g}_i\big|^2.\label{ipcMSE}
\end{align}
In analogy with the linear case we wish to
\begin{subequations}
\begin{align}
&\underset{\boldsymbol{v}_k}{\operatorname{\text{minimize}}}~~\cev{\varepsilon}_{\text{ipc}}(k,\gamma)\\
&\text{subject to}~~\|\boldsymbol{v}_k\|^2=P_k.\label{fullpow}
\end{align}
\end{subequations}
The solution to this problem is given by (\ref{eq:upNlnFilter2}).
Due to the constraint (\ref{fullpow}),
the term $\sigma_k^2\|\boldsymbol{v}_k\|^2$ in (\ref{ipcMSE}) is a constant
so that the solution to this problem again does not depend on
the noise level $\sigma_k^2$ at user $k$.

In the absence of CSI, the previous iterative optimization
can be implemented by iterating the following steps:

\begin{figure*}
  \center
  \includegraphics[width=0.65\textwidth]{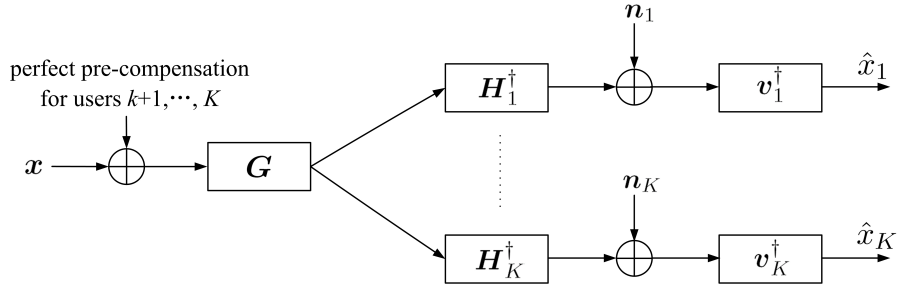}
  \caption{\small Backward training of the filter at user $k$ with perfect interference pre-compensation for users $k+1$ to $K$ in the downlink.}
   \label{backward_training}
\end{figure*}

\begin{enumerate}
\item {\em Forward training:} The $K$ users synchronously transmit uplink training symbols
using the $\boldsymbol{v}_k$'s as uplink precoders,
and the BTS optimizes the feedforward filters $\{\boldsymbol{g}_k\}$ with
perfect interference cancellation.
\item {\em Backward training:} The BTS transmits training symbols using
$\{\boldsymbol{g}_k\}$ as downlink precoders with perfect interference pre-compensation
in the {\em reverse} user order as for the uplink.
User $k$ optimizes its receive filter $\boldsymbol{v}_k$ subject to a norm constraint.
\item Iterate steps 1 and 2.
\end{enumerate}

Downlink interference pre-compensation can be achieved with dirty paper coding.
However, this has high computational complexity.
Here we will instead consider two alternative approaches for
pre-compensating interference at the BTS during the training phase:
sequential training and Tomlinson-Harashima precoding (THP) \cite{TseVis05}.
THP for MIMO channels is described in \cite{WinFis04TWC}
and can be used to pre-compensate downlink interference
in a symbol-by-symbol manner. Specifically,
the DFD used in the uplink can be reused as the precoders for THP.

As before, we replace the MSE optimization criterion with a least squares criterion
in adaptive mode.
In the forward training step the users synchronously transmit sequences of
$\NT$ training symbols given by the matrix
$\vec{\boldsymbol{S}}^T$ where
$\vec{\boldsymbol{S}}\!=\![\vec{\boldsymbol{s}}_1^T,\ldots,
\vec{\boldsymbol{s}}_K^T]$ and $\vec{\boldsymbol{s}}_k$
is the $1\times\NT$ row vector containing the training symbols.
The received signal at the BTS is then given by (\ref{upsignal})
with $\vec{x}_k\!=\!\vec{s}_k$.
To obtain the feedforward and feedback receive filters for user $k$,
the known training symbols for users $1,\dots,k-1$ can be included
in the least squares objective as proposed in \cite{HonWoo04TWC}.
The corresponding sequence of estimated symbols can be written as
\begin{equation}
\hat{\boldsymbol{s}}_k=\boldsymbol{g}_k^\dagger\vec{\boldsymbol{Y}}-\boldsymbol{b}_k\vec{\boldsymbol{S}}_1^{k-1},\label{estAlpha}
\end{equation}
where $\vec{\boldsymbol{Y}}\!=\!\big[\vec{\boldsymbol{y}}(1),\ldots,\vec{\boldsymbol{y}}(\NT)\big]$, $\boldsymbol{b}_k$ is a row vector of dimension $k-1$ that contains the nonzero elements of the $k$-th row of $\boldsymbol{B}\!-\!\boldsymbol{I}$, and $\vec{\boldsymbol{S}}_1^{k-1}\!=\![\vec{\boldsymbol{s}}_1^T,\ldots,\vec{\boldsymbol{s}}_{k-1}^T]^T$.
The receive filters are then chosen to minimize the cost function $\|\vec{\boldsymbol{s}}_k\!-\!\boldsymbol{g}_k^\dagger\vec{\boldsymbol{Y}}\!+\!\boldsymbol{b}_k\vec{\boldsymbol{S}}_1^{k-1}\|^2$, and the corresponding solution is given by
\begin{equation}
\begin{bmatrix}\boldsymbol{g}_k\\-\boldsymbol{b}_k^\dagger\end{bmatrix}
\!=\!\begin{bmatrix}\vec{\boldsymbol{Y}}\vec{\boldsymbol{Y}}^\dagger & \vec{\boldsymbol{Y}}(\vec{\boldsymbol{S}}_1^{k-1})^\dagger\\ \vec{\boldsymbol{S}}_1^{k-1}\vec{\boldsymbol{Y}}^\dagger & \vec{\boldsymbol{S}}_1^{k-1}(\vec{\boldsymbol{S}}_1^{k-1})^\dagger\end{bmatrix}^{-1}
\!\begin{bmatrix}\vec{\boldsymbol{Y}}\\\vec{\boldsymbol{S}}_1^{k-1}\end{bmatrix}\vec{\boldsymbol{s}}_k^\dagger.
\end{equation}
We emphasize that in this way the filters are estimated {\em directly}, as opposed
to estimating the channels first followed by computation of the filters.

We next describe two downlink training methods that can be used to estimate
the uplink precoders, accounting for the interference cancellation at the BTS.

\emph{Sequential training:} Training of the receive filter at each user
can be implemented by sequentially scheduling transmission
of training symbols to the users in order $k=1,\dots,K$.
Specifically, we layer the training sequences for user 1, then add user 2,
then add user 3, and so on, until user $K$.
To optimize the receiver filter at user $k$,
training symbols for users $1,\ldots,k$ are therefore simultaneously transmitted
with linear precoding filters $\boldsymbol{g}_1 , \dots,\boldsymbol{g}_k$
so that the received signal vector at user $k$
is given by (\ref{eq:downrecSigNln}) with
$\boldsymbol{t}_i\!=\!\boldsymbol{g}_i$ and
$\cev{x_i}\!=\!\cev{s_i}$
(using the notation defined in Section \ref{sec:LinearTrain}).
As for the linear case, user $k$ can estimate its receive filter
$\boldsymbol{v}_k$ by minimizing the least squares cost function
$\|\cev{\boldsymbol{s}_k}\!-\!\boldsymbol{v}_k^\dagger\cev{\boldsymbol{Y}_k}\|^2/\gamma\!+\!\mu_k\big(\|\boldsymbol{v}_k\|^2\!-\!P_k\big)$
and the corresponding solution is the same as in (\ref{LSuser}).

This sequential training method perfectly pre-compensates for interference.
However, a drawback is that when the filter for user $k$ is being estimated,
the symbols corresponding to users $1,\dots,k-1$ are needed only to generate
the appropriate interference; they are not used to update the corresponding filters.
Hence this method will generally require more training than the simultaneous training
method described next.

\emph{Training with THP:}
This method is illustrated in Fig.~\ref{downlink}, which shows the
uplink DFD filters being used as the precoders for THP.
This can be viewed as moving the feedback part of the DFD
in Fig.~\ref{uplink} to the transmitter with the slicer function
replaced by a modulo operation, which depends on the signal constellation
used for modulation (see also \cite{WinFis04TWC}).
\begin{figure*}
  \center
  \includegraphics[width=0.75\textwidth]{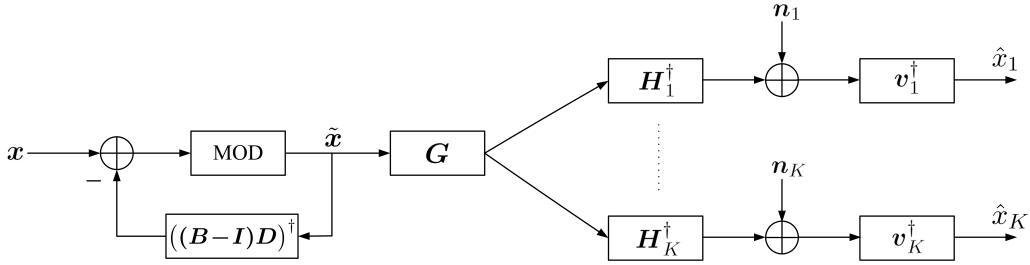}
  \caption{\small Tomlinson-Harashima precoding (THP) in the downlink maps $x$ to $\tilde{x}$.}
   \label{downlink}
\end{figure*}

Because the estimates of the symbols in the uplink are biased,
the DFD feedback filter must be modified to remove this bias
when used in the THP-precoder. Specifically, the feedback filter in
the downlink precoder is the Hermitian transpose of the matrix
\begin{equation}\notag
(\boldsymbol{B}-\boldsymbol{I})\boldsymbol{D}=
\begin{bmatrix}
0 & \hdotsfor{3} & 0\\
\dfrac{b_{21}}{\alpha_1} & 0 & \hdotsfor{2} & 0\\
\dfrac{b_{31}}{\alpha_1} & \dfrac{b_{32}}{\alpha_2} & 0 & \dots & 0\\
\hdotsfor{5}\\
\dfrac{b_{K1}}{\alpha_1} &
\dfrac{b_{K2}}{\alpha_2} &
\dfrac{b_{K3}}{\alpha_3} & \dots & 0
\end{bmatrix}
\end{equation}
where $\boldsymbol{D}\!=\!\text{diag}\big(\frac{1}{\alpha_1},\ldots,\frac{1}{\alpha_K}\big)$
and $\alpha_k\!\triangleq\!\boldsymbol{g}_k^\dagger\boldsymbol{H}_k\boldsymbol{v}_k$
is the bias factor for estimating $\vec{x}_k$ in (\ref{dfesignal}).

We now describe the interference pre-compensation for the downlink
using the feedforward receive filters $\{\boldsymbol{g}_k\}$ as precoders
and $\boldsymbol{v}_k$ (the precoder for user $k$ in the uplink)
as the receiver filter at user $k$.
Denoting the input to the THP encoder as $\cev{x_k}$,
the output is then
\begin{align}
\tilde{x}_{k}
&=f_{\text{mod}}\bigg(\cev{x_k}-\sum_{i=k+1}^K\bigg(\frac{b_{ik}}{\alpha_k}\bigg)^\dagger\tilde{x}_i\bigg)\\
&=Z_k+\cev{x_k}-\sum_{i=k+1}^K\bigg(\frac{b_{ik}}{\alpha_k}\bigg)^\dagger\tilde{x}_i,\label{modoper}
\end{align}
where $f_{\text{mod}}(\cdot)$ denotes the modulo operation and
$Z_k$ denotes the corresponding offset due to the modulo operation,
which depends on the current symbol $\cev{x_k}$,
the previous encoded symbols $\tilde{x}_i,i\!=\!k\!+\!1,\ldots,K$,
and the feedback matrix. In general, each real and imaginary part
is an integer multiple of the maximum distance
between two constellation points (see \cite{WinFis04TWC}).

The received symbol at user $k$ at the output of the receive filter is given by
\begin{align}
\hat{x}_k
&=\boldsymbol{v}_k^\dagger\bigg(\boldsymbol{H}_k^\dagger\bigg(\sum_{i=1}^{K}\boldsymbol{g}_i\tilde{x}_i\bigg)+\cev{\boldsymbol{n}}_k\bigg)\\
&=\boldsymbol{v}_k^\dagger\boldsymbol{H}_k^\dagger\boldsymbol{g}_k(Z_k+\cev{x_k})+\boldsymbol{v}_k^\dagger\bigg(\boldsymbol{H}_k^\dagger\bigg(\sum_{i=1}^{k-1}\boldsymbol{g}_i\tilde{x}_i\bigg)+\cev{\boldsymbol{n}}_k\bigg)\notag\\
&~~~~~~+\sum_{i=k+1}^{K}\bigg(\boldsymbol{v}_k^\dagger\boldsymbol{H}_k^\dagger\boldsymbol{g}_i\tilde{x}_i-\boldsymbol{v}_k^\dagger\boldsymbol{H}_k^\dagger\boldsymbol{g}_k\bigg(\frac{b_{ik}}{\alpha_k}\bigg)^\dagger \tilde{x}_i\bigg).\label{drecsignal}
\end{align}
Since $b_{ik}\!=\!\boldsymbol{g}_i^\dagger\boldsymbol{H}_k\boldsymbol{v}_k$ and $\alpha_k\!=\!\boldsymbol{g}_k^\dagger\boldsymbol{H}_k\boldsymbol{v}_k$, the last summation in (\ref{drecsignal}) is zero, i.e., the interference to user $k$ generated by the users $k\!+\!1$ to $K$  can be completely removed, which means that user $k$ experiences interference only from users $1$ through $k\!-\!1$. The effective desired symbol is then $Z_k\!+\!\cev{x_k}$.

We emphasize that THP in the downlink is implemented in the \emph{reverse}
user order as cancellation in the uplink. With the modulo operation,
the power of the encoded symbols may increase relative to the original power (prior to THP)
to a small extent, and the increment decreases as the constellation size increases \cite{WinFis04TWC,SchShi05VTC}. In (\ref{drecsignal}), since the powers of the effective training symbols and the residual interference are different from the corresponding powers in (\ref{eq:downrecSigNln}) (with $\boldsymbol{t}_i\!=\!\boldsymbol{g}_i$), training with THP in the downlink
is an approximate approach to estimating (\ref{eq:upNlnFilter2}).

The training symbols for the downlink for different users are
precoded according to (\ref{modoper}), where the bias factor
for $\vec{s}_k$ can be estimated as
$\alpha_k\!=\!\hat{\boldsymbol{s}}_k\vec{\boldsymbol{s}}_k^\dagger/n$
and $\hat{\boldsymbol{s}}_k$ is given by (\ref{estAlpha})
once the uplink receive filters have been obtained.
The input to the feedforward filters $\boldsymbol{g}_k$'s (used as precoders)
at the BTS is then given by (\ref{modoper}) with $\cev{x_i}$
replaced by the training symbol $\cev{s_i}$ and
$\tilde{x}_i$ replaced by $\tilde{s}_i$.
Taking the desired symbol in (\ref{drecsignal}) as $Z_k\!+\!\cev{s_k}$,
the receiver filter is then obtained by minimizing the least squares cost function
$\|\boldsymbol{Z}_k\!+\!\cev{\boldsymbol{s}_k}\!-\!\boldsymbol{v}_k^\dagger\cev{\boldsymbol{Y}_k}\|^2/\gamma\!+\!\mu_k\big(\|\boldsymbol{v}_k\|^2\!-\!P_k\big)$, where $\cev{\boldsymbol{Y}_k}\!=\!\big[\cev{\boldsymbol{y}_k}(1),\ldots,\cev{\boldsymbol{y}_k}(n)\big]$ denotes the received signal vectors and each entry $\cev{\boldsymbol{y}_k}$ is given by (\ref{eq:downrecSig}) with $\cev{x_i}\!=\!\tilde{s}_i$ and $\boldsymbol{t}_i\!=\!\boldsymbol{g}_i$. The corresponding solution is given by
\begin{equation}
\boldsymbol{v}_k=\big(\cev{\boldsymbol{Y}_k}\cev{\boldsymbol{Y}_k}^\dagger+\mu_k\boldsymbol{I}\big)^{-1}\cev{\boldsymbol{Y}_k}\big(\boldsymbol{Z}_k+\cev{\boldsymbol{s}_k}\big)^\dagger,\label{LSthpuser}
\end{equation}
where $\mu_k$ is chosen to satisfy the constraint $\|\boldsymbol{v}_k\|^2\!=\!P_k$.

This expression for $\boldsymbol{v}_k$ depends on the offset $Z_k$, which is introduced
at the BTS, and therefore must be estimated at the receivers.
Here we detect $Z_k$ first before estimating $\boldsymbol{v}_k$.
For downlink training the detection can be based on the output
of the filter $\boldsymbol{v}_k$ (uplink precoder), which
from (\ref{drecsignal}) is given by
\begin{align}
\tilde{Z}_k\!\!=\!\boldsymbol{v}_k^\dagger\boldsymbol{H}_k^\dagger\boldsymbol{g}_k
\!\big(Z_k\!\!+\!\!\cev{s_k}\big)
\!+\!\boldsymbol{v}_k^\dagger\!\bigg(\sum_{i=1}^{k-1}\!\boldsymbol{H}_k^\dagger\boldsymbol{g}_i\tilde{s}_i
\!+\!\cev{\boldsymbol{n}}_k\!\bigg).\label{effetrsym}
\end{align}
The estimate of $Z_k$ can be obtained by slicing $\tilde{Z}_k$.
Since the points in the discrete set of possible $Z_k$'s have a minimum distance
that is greater than the maximum distance between constellation points
in each real or imaginary dimension, we expect that the performance loss
due to errors in detecting $Z_k$ is small.
This is verified by the numerical examples in the next section.

\section{Numerical Results}
\label{sec:numericalResults}
\subsection{Uplink Examples}
Fig.~5(a) shows uplink sum rates versus SNR achieved by different schemes
in a single cell with $K\!=\!4$ users.
There are $N\!=\!4$ antennas at the BTS and two antennas
at each user $(N_k\!=\!2,~k=1,\ldots,K)$.
The sum rate curve is obtained by choosing the filters
$\boldsymbol{v}_k$'s to maximize the achievable sum rate
$\log\det\big(\boldsymbol{I}_N\!+\!\sum_{k=1}^{K}\!
\boldsymbol{H}_k\boldsymbol{v}_k\boldsymbol{v}_k^\dagger\boldsymbol{H}_k^\dagger\big)$.
The other four curves show the performance of linear and nonlinear filters
with 100 forward-backward iterations, and with only two iterations.
In each iteration an MMSE update is performed, corresponding
to an infinite amount of training.
Equal transmit power is assumed for all the users and
the curves are averaged over a thousand random channel realizations.
For the nonlinear precoder, THP with quadrature phase-shift keying (QPSK)
is used for downlink training.

The results show that with 100 forward-backward iterations, the sum rate
with nonlinear filters is quite close to the uplink sum capacity, and is somewhat
higher than with linear filters. With two iterations the gap between
linear and nonlinear filters increases, especially at high SNRs.
Hence nonlinear filters are more attractive when the number of forward-backward
iterations is highly constrained.

Fig.~5(b) compares the sum rates achieved by limited and unlimited bi-directional training.
Two forward-backward iterations are used and the training sequence
for each user is a randomly generated QPSK sequence.
With 20 training symbols per iteration,
the performance is close to that with MSE updates (infinite amount of training).
This figure indicates that nonlinear filters still outperform linear filters
with finite training. However, the gap in sum rate for nonlinear filters
with infinite and finite training
is slightly larger compared to the analogous gap with linear filters.
Fig.~5(c) shows the effects of varying the number of training symbols on nonlinear filters,
again with two iterations.
\begin{figure}
  \setcounter{figure}{5}
  \center
  \subfloat[]{
  \includegraphics[width=0.5\textwidth]{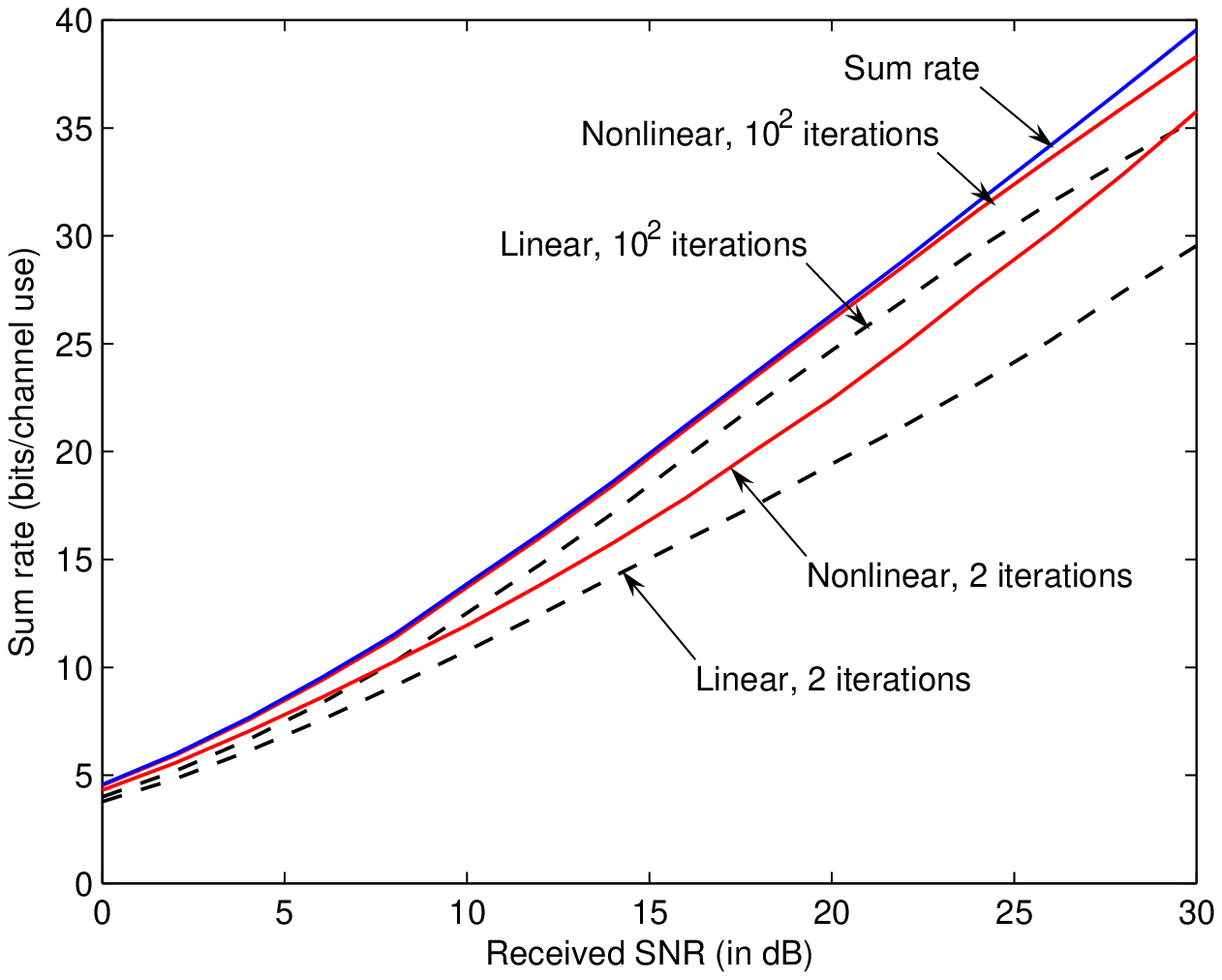}
  \label{f:up_sum_rate_2iters}
} \,
  \subfloat[]{
  \includegraphics[width=0.5\textwidth]{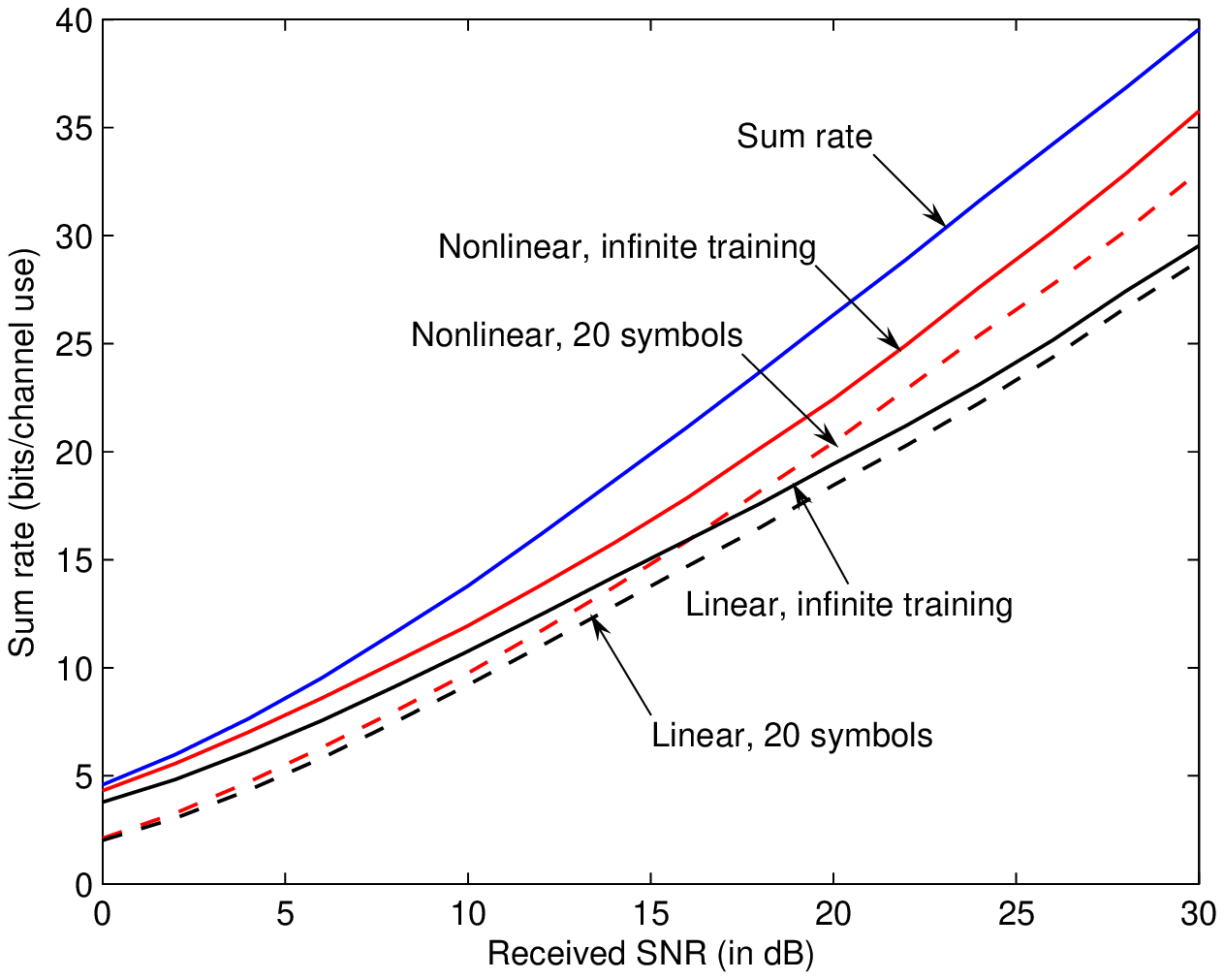}
  \label{f:up_sum_rate_2iters_20symbols}
} \,
\end{figure}
\begin{figure}
  \setcounter{figure}{5}
  \ContinuedFloat
  \center
  \subfloat[]{
  \includegraphics[width=0.5\textwidth]{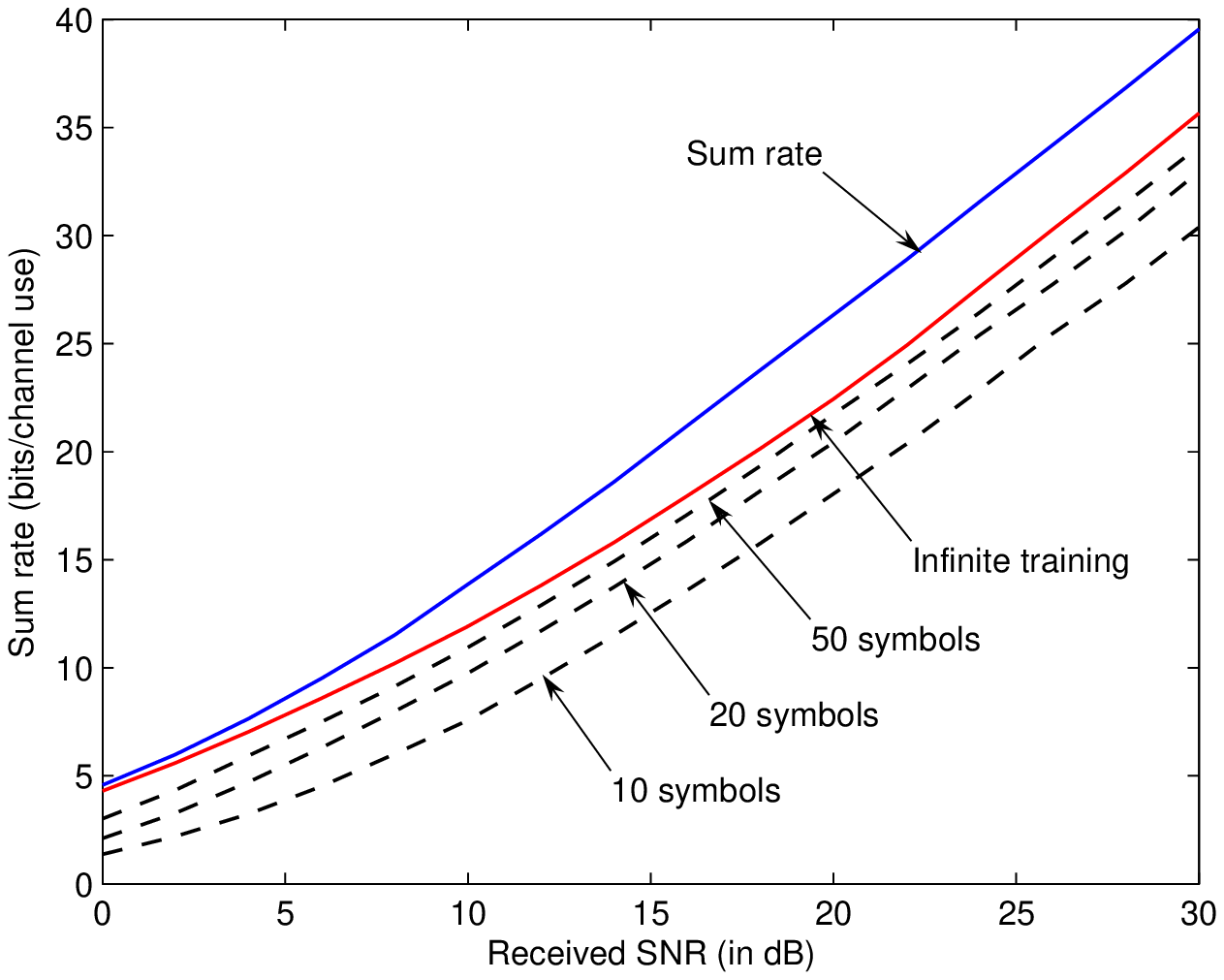}
  \label{f:up_sum_rate_2iters_training}
}
\caption{\small Uplink sum rate versus SNR with $K\!=\!4$ users, $\Nb\!=\!4$ antennas
at the BTS, and $N_k\!=\!2$ antennas at each user.
There are two forward-backward iterations in (b) and (c).}
\label{f:up_sum_rate}
\end{figure}

Fig.~\ref{f:up_sum_rate_3Nr} shows the achievable sum rate
in an overloaded system with $N\!=\!3$ receive antennas
at the BTS and $K\!=\!4$ users. Each iteration is an MMSE update (infinite training).
With linear filters the sum rate saturates as the received SNR increases
since there are not enough degrees of freedom to null the interference.
However, with nonlinear filters the maximum degrees of freedom
(asymptotic slope of the sum rate) of three can be achieved
as the received SNR increases.\footnote{The numerical results also show
that in the uplink each user gets a positive rate although the rates differ
across users. In the downlink each user gets roughly the same rate.}
The sum rate approaches the capacity upper bound as the number of
forward-backward iterations increases.
(More forward-backward iterations are needed at high SNRs
to approach the sum capacity.) With three receive antennas
and linear filters, full degrees of freedom can also be achieved
if the BTS schedules transmissions from only three out of the four users.
The corresponding achievable sum rate is also included
for comparison where the three users are scheduled randomly.
It is observed that nonlinear filters (with four users scheduled)
still outperform linear filters over a wide range of SNRs
in this scenario. However, at high SNRs (above 22 dB)
the linear filters outperform the nonlinear filters,
since the nonlinear filters require more iterations to reach the sum capacity.


%
%
\begin{figure}
  \center
  \includegraphics[width=0.5\textwidth]{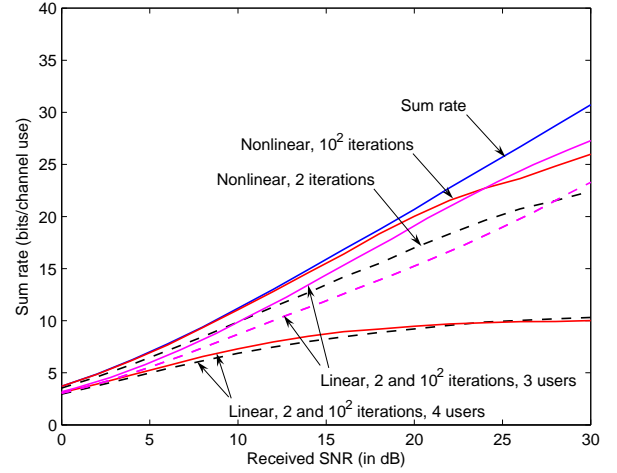}
\caption{\small Uplink sum rate with $\Nb\!=\!3$ antennas at the BTS
and $N_k\!=\!2$ antennas at each user. The curves with nonlinear filtering
correspond to $K\!=\!4$ users.}
\label{f:up_sum_rate_3Nr}
\end{figure}

\subsection{Simultaneous Uplink/Downlink Optimization}
Fig.~\ref{f:sum_rate_unified} compares the uplink and downlink sum rates
achieved by separately minimizing the MSE for each
with the corresponding rates achieved with simultaneous
optimization, as described in Section \ref{Sec:uniApproach}.
For simultaneous optimization the transmit power at each node
is normalized to satisfy the power constraint in each forward-backward iteration.
While there are slight differences in the curves,
simultaneous optimization performs essentially the same as separately optimizing
the uplink or downlink precoders and receive filters.
Therefore the simultaneous approach is preferable from the perspective
of reducing the training overhead.
\begin{figure}
  \center
  \includegraphics[width=0.5\textwidth]{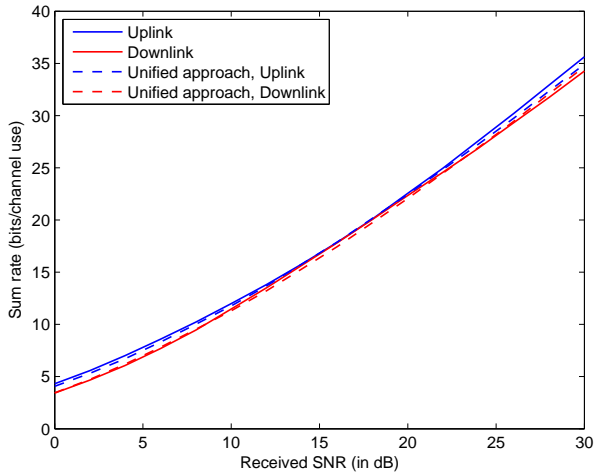}
\caption{\small Comparison of simultaneous training with separate
uplink/downlink training. Uplink and downlink sum rates are shown
with nonlinear filters, $N\!=\!4$ antennas at the BTS,
$N_k\!=\!2$ antennas at each user, and two forward-backward iterations.}
\label{f:sum_rate_unified}
\end{figure}

\subsection{Multiple-Cell Scenario}
It is straightforward to extend bi-directional training
for linear and nonlinear filters discussed in Section \ref{s:training}
to the scenario with multiple cells.
Joint optimization of linear precoders and receivers in a cellular system
has been considered in many references, including
\cite{CodTol07TSP, SuhTse11TC, ZhuBer11ICASSP, KomTol13TSP,HanLeN13TSP}.
Achieving the maximum degrees of freedom in the multiple-cell scenario
requires interference alignment across cells, as discussed in
\cite{GomCad11IT,ZhuBer11ICASSP,SuhTse11TC}. That is, the other-cell interference
lies in a subspace with lower dimension than the number of other-cell users.
A related bi-directional training method was presented in \cite{ZhuBer11ICASSP}.

Here we apply a nonlinear DFD at each BTS for the uplink to cancel the interference
from intra-cell users, while mitigating inter-cell interference via the
precoders and feedforward filters. Similarly, for the downlink we can
apply THP as described in the preceding section to precompensate
for the intra-cell interference while avoiding inter-cell interference.
The training procedures are the same as described in Section \ref{sec:nonlinear},
where the users across cells synchronously transmit on the uplink
and the BTS's synchronously transmit on the downlink.

The algorithms and results in Sections \ref{sec:linearup}--\ref{s:training}
can be easily extended to the multiple-cell scenario.
Fig.~\ref{f:sum_rate} shows uplink sum rate versus SNR in a two-cell scenario
with three antennas at each BTS and two dual-antenna users in each cell.
All of the direct- and cross-channels are assumed to be of equal strength.
The necessary conditions for interference alignment presented
in \cite{ZhuBer11ICASSP} are satisfied with equality with these parameters.
In this scenario, the inter-cell interference in each cell can be aligned in
a one-dimensional subspace \cite{CadJaf08IT}, so that two degrees of freedom
can be achieved in each cell.

Fig.~8(a) compares sum rates for linear and nonlinear filters with
different numbers of iterations. In the high-SNR regime there are
substantial differences among the results with two, 100, and 1000 forward-backward iterations,
which indicates that it takes many more iterations for the precoders and filters
to converge than for the single-cell scenario.
That is, Fig.~5(a) shows that for the single-cell scenario,
the achievable sum rate with two iterations is comparable
to that with 100 iterations. In Fig.~8(b), the performance of
finite bi-directional training is presented with two forward-backward iterations.
For these results downlink training with nonlinear filters uses THP with QPSK.
With 20 training symbols per iteration, the performance is close to that
with MMSE updates (infinite training).
Nonlinear filters give a noticeable gain over linear filters
with either infinite or finite training over a wide range of SNRs.
\begin{figure}
  \center
  \subfloat[]{
  \includegraphics[width=0.5\textwidth]{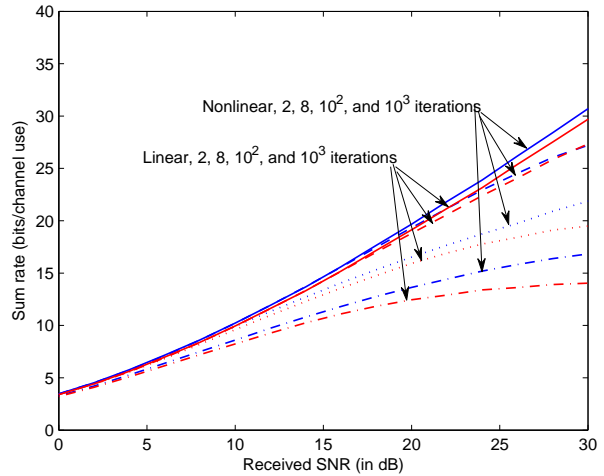}
  \label{f:sum_rate_2iters}
} \,
  \subfloat[]{
  \includegraphics[width=0.5\textwidth]{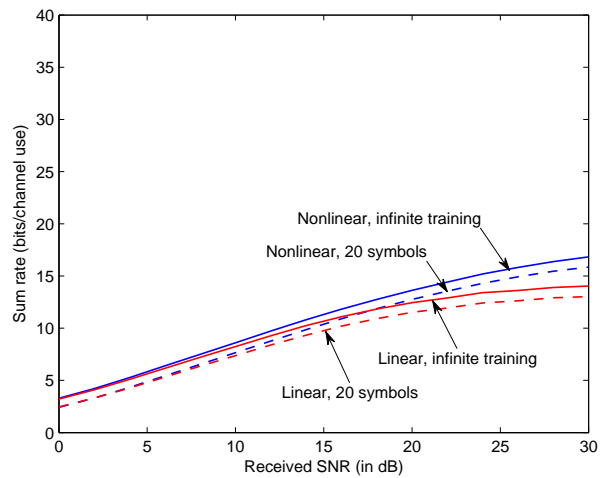}
  \label{f:sum_rate_2iters_20symbols}
} \,
\caption{\small Uplink sum rate per cell with two cells, $K\!=\!2$ users in each cell,
$\Nb\!=\!3$ antennas at each BTS, and $N_k\!=\!2$ antennas at each user.
There are two forward-backward iterations in (b).}
\label{f:sum_rate}
\end{figure}

Fig.~{\ref{f:sum_rate_3users}} shows the achievable uplink sum rate
with two cells, three users in each cell, and three antennas at each user.
Here the necessary conditions for interference alignment
in \cite{ZhuBer11ICASSP} are not satisfied, so that interference alignment
is not achievable. With linear filters the sum rate saturates as
the SNR increases irrespective of the number of forward-backward iterations.
In contrast, with nonlinear filters the sum rate increases with SNR
given enough forward-backward iterations (so that the filters can converge),
and the maximum of two degrees of freedom per cell can be achieved.
(As in the last example, other-cell interference occupies a subspace
of one dimension leaving two dimensions for the desired signals
\cite{SuhTse11TC}.)
\begin{figure}
  \center
  \includegraphics[width=0.5\textwidth]{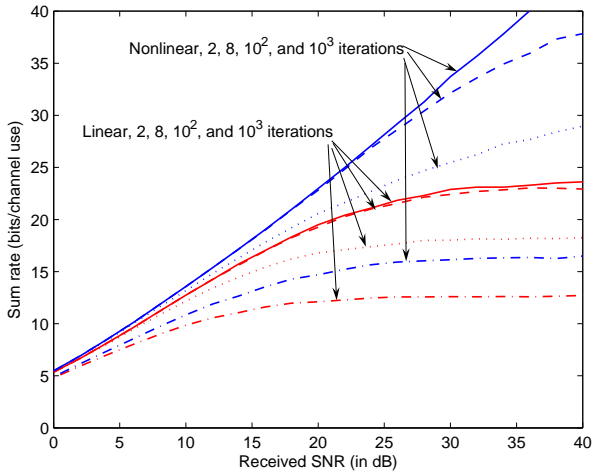}
\caption{\small Uplink sum rate per cell
with two cells,
$K\!=\!3$ users in each cell, $N\!=\!3$ antennas at each BTS, and three antennas at each user.}
\label{f:sum_rate_3users}
\end{figure}

\section{Conclusions}
A distributed iterative bi-directional training approach has been proposed for jointly optimizing
the filters in base transceiver stations and user equipments in cellular systems.
The approach allows for a nonlinear decision feedback canceller in the uplink
and a feedback precoder for downlink interference pre-cancellation.
In the absence of channel state information,
the filters can be directly estimated by synchronously transmitting uplink pilots
from all users and downlink pilots from all BTS's.
Downlink training may be nonlinearly precoded as well to enable
simultaneous training of all user filters (uplink precoders).
This adaptive scheme avoids the need to stagger pilots across cells,
and to exchange channel estimates across users and BTSs.

With nonlinear filters, it is observed that the uplink sum rate
achieved by bi-directional training in a single cell approaches
the sum capacity with increasing number of iterations.
Although there are more parameters to be estimated compared to linear filters,
with nonlinear filters finite training can significantly
outperform linear filters with a limited number of forward-backward iterations.
Furthermore, the nonlinear filters have been observed to achieve
the available degrees of freedom even when interference alignment
is infeasible with linear precoders and filters.

Remaining open issues are how to trade off the number of iterations
and the amount of training per iteration, and how to determine the optimum number
of data-streams for each user if multiple data-streams are allowed.
The proof of convergence for the simultaneous approach to uplink/downlink
filter optimization in which the filters are normalized at each step
also remains open. Finally, an important issue in practice is how to combine
such joint beamforming schemes with the scheduling of users
across time-frequency resources.

\section*{Appendix: Proof of Proposition 1}
It is enough to show that given any set of precoders $\{\boldsymbol{v}_k\}$,
to achieve the minimum uplink sum MMSE each user should transmit with full power.
Denote $\boldsymbol{H}\!=\![\boldsymbol{H}_1\!\boldsymbol{v}_1,\ldots,\boldsymbol{H}_K\!\boldsymbol{v}_K]$, $\vec{\boldsymbol{x}}\!=\![\vec{x}_1,\ldots,\vec{x}_K]$, and $\boldsymbol{\hat{x}}\!=\![\hat{x}_1,\ldots,\hat{x}_K]$.
The covariance matrix for the MMSE estimation error
$\boldsymbol{\epsilon}\!=\!\hat{\boldsymbol{x}}\!-\!\vec{\boldsymbol{x}}$
is given by
\begin{equation}
\boldsymbol{C}_{\boldsymbol{\epsilon}}
=\boldsymbol{I}_K-\boldsymbol{H}^\dagger(\boldsymbol{H}\boldsymbol{H}^\dagger+\boldsymbol{I}_N)^{-1}\boldsymbol{H}.\label{MMSE}
\end{equation}
Taking the trace of both sides gives the sum MSE
\begin{align}
\text{\large{trace}}\big(\boldsymbol{C}_{\boldsymbol{\epsilon}}\big)
&=K-\text{\large{trace}}\big(\boldsymbol{H}\boldsymbol{H}^\dagger(\boldsymbol{H}\boldsymbol{H}^\dagger+\sigma^2\boldsymbol{I}_N)^{-1}\big)\label{traceeq}\\
&=K-\sum_{i=1}^{N}\frac{\lambda_i}{\lambda_i+\sigma^2}\label{totalMSE}
\end{align}
where the $\lambda_i$'s are the eigenvalues of $\boldsymbol{H}\boldsymbol{H}^\dagger$.
Assume that the transmit power of user $k$ is scaled up by $\beta_k^2\!\geq\!1$,
which implies the norm of the $k$-th column of $\boldsymbol{H}$ is scaled by $\beta_k$.
Define the new channel matrix as $\tilde{\boldsymbol{H}}\!=\![\beta_1\!\boldsymbol{H}_1\!\boldsymbol{v}_1,\ldots,\beta_K\!\boldsymbol{H}_K\!\boldsymbol{v}_K]$.
Denote $\tilde{\lambda}_k$ as the $k$-th largest eigenvalue of
$\tilde{\boldsymbol{H}}\tilde{\boldsymbol{H}}^\dagger$.
Since
$\tilde{\boldsymbol{H}}\tilde{\boldsymbol{H}}^\dagger\!=
\!\boldsymbol{H}\boldsymbol{H}^\dagger\!+\!\sum_{k=1}^n
(\beta_k^2\!-\!1)\boldsymbol{H}_k\!\boldsymbol{v}_k\!\boldsymbol{v}_k^\dagger
\!\boldsymbol{H}_k^\dagger$ and the matrix
$\sum_{k=1}^n(\beta_k^2\!-\!1)\boldsymbol{H}_k
\!\boldsymbol{v}_k\!\boldsymbol{v}_k^\dagger\!\boldsymbol{H}_k^\dagger$
is positive semidefinite, by \cite[Corollary 4.3.3]{HorCha95},
we must have $\tilde{\lambda}_i\!\geq\!\lambda_i$ for each $i$.
Hence increasing the power for any user must decrease the sum MSE.

\end{document}